\documentclass[a4paper,fleqn]{cas-sc}

\usepackage[numbers]{natbib}

\usepackage{color}
\usepackage{times}
\usepackage{epsfig}
\usepackage{amsthm}
\usepackage{amsmath}
\usepackage{algorithm}
\usepackage{algorithmic}
\usepackage{mathrsfs}
\usepackage{epstopdf}
\usepackage{subcaption}
\usepackage{float}

\newtheorem{ass}{Assumption}

\newtheorem{definition}{Definition}
\newtheorem{thm}{Theorem}

\newtheorem{rem}{Remark}

\def\tsc#1{\csdef{#1}{\textsc{\lowercase{#1}}\xspace}}
\tsc{WGM}
\tsc{QE}
\tsc{EP}
\tsc{PMS}
\tsc{BEC}
\tsc{DE}

\begin{document}
\let\WriteBookmarks\relax
\def\floatpagepagefraction{1}
\def\textpagefraction{.001}
\shorttitle{Bifurcation Beyond Surface-Tangential Asymptotic Convergence}
\shortauthors{Fan Zhang et~al.}

\title[mode = title]{Bifurcation Beyond Surface-Tangential Asymptotic Convergence in Continuous Sliding Mode Control}               

\author[1,2]{Fan Zhang}[orcid=0000-0003-3664-3553]
\ead{zhangfan_nwpu@nwpu.edu.cn}
\credit{Conceptualization, Funding acquisition, Methodology, Writing - original draft}

\affiliation[1]{organization={School of Astronautics, Northwestern Polytechnical University, }, 
                city={Xi'an},
                postcode={710072}, 
                country={China}}

\author[1]{Jingwen Xu}[orcid = 0009-0008-7174-437X]
\ead{xjw0324@nwpu.edu.cn}
\credit{Writing - Review \& editing}

\author[1]{Peng Li}[orcid=0000-0002-4895-7659]
\ead{lp@nwpu.edu.cn}
\credit{Writing - Review \& editing}

\author[1,2]{Jun Zhou}[orcid=0000-0001-5810-5153]
\ead{zhoujun@nwpu.edu.cn}
\credit{Validation}

\author[1]{Yaohua Guo}[orcid=0000-0002-1489-969X]
\cormark[1]
\ead{y.guo2@nwpu.edu.cn}
\credit{Funding acquisition, Writing - review \& editing}

\affiliation[2]{organization={Ningbo Institute of Northwestern Polytechnical University},
                postcode={315103}, 
                city={Ningbo},
                country={China}}

\cortext[cor1]{Corresponding author}

\begin{abstract}
For second-order systems under continuous sliding mode control (SMC), the literature has long relied, largely through phase-portrait illustrations, on the implicit convention that the phase-plane trajectory approaches the equilibrium along a direction tangential to the designed sliding surface. This paper investigates this tangential convergence assumption through a rigorous phase-plane analysis of the double-integrator system subject to standard linear SMC. We reveal that the asymptotic state ratio is not unique but instead exhibits a bifurcation that depends on the control gains, the sliding surface parameter, and the initial conditions. The key finding is that the system may converge along an implicit secondary manifold rather than aligning with the designed sliding surface, implying that smooth entry is not globally guaranteed. We derive closed-form expressions for the convergence ratios and establish a classification framework that precisely characterizes when the smooth-entry assumption holds and when it fails. The analysis is further extended to classical PD control, and we show that the bifurcation threshold coincides with the critical damping boundary that separates the two convergence regimes. These findings bridge terminal geometry and convergence smoothness, providing a predictive framework for high-performance motion control design. Simulation results validate the proposed classification of convergence regimes.
\end{abstract}




\maketitle

\section{Introduction}

Sliding mode control (SMC) has long been applied to the operation of classical second-order systems, owing to its conceptual simplicity and inherent robustness against uncertainties \citep{utkin2013, JAVADI2026108616}. The core idea is to enforce the system states onto a prescribed sliding manifold and maintain it there thereafter \citep{young1999, OVALLE2026112797, KRAVARIS20121583}. Once the sliding motion is established, the system behavior is entirely governed by the reduced-order dynamics on this manifold, which are independent of plant perturbations. Over the past decades, research in this field has reached a high level of maturity, with major thrusts concentrated on three pillars: achieving finite/fixed-time convergence \citep{perruquetti2002, Chang2024}, enhancing robustness against matched/unmatched uncertainties \citep{rm24, bm2025}, and alleviating the chattering issue \citep{levant1993, deng2024, JIANG2026106439}.

Despite these extensive developments, a surprisingly fundamental geometric aspect concerning \textit{the terminal convergence behavior of continuous SMC systems} has remained largely overlooked. For classical second-order systems, the literature has long adhered, largely through the phase-portrait illustrations \citep{lfs2018, mfbc2018}, to an implicit convention that the phase-plane trajectory enters the equilibrium along the sliding surface \( s=0 \), namely \textit{smoothly}. Specifically, it is tacitly presumed that the trajectory approaches and crosses the manifold with its phase-plane tangent aligned with the sliding surface, so that the ratio of the two state variables at the instant of entry coincides with the slope determined by the manifold. To the best of our knowledge, this smooth-entry hypothesis has never been formally investigated in prior works. One likely reason is that the entry detail is deemed immaterial to the standard reachability and stability analyses: for the purpose of proving convergence, whether this geometric alignment holds globally does not affect the validity of the conventional arguments. However, the lack of scrutiny becomes consequential when one moves beyond qualitative stability to finite/fixed-time convergence behavior. As we shall demonstrate in the subsequent sections, the validity of the smooth-entry assumption has a direct and nontrivial bearing on whether the convergence time is finite or infinite, thereby challenging a foundational premise that underlies the standard transient analysis and convergence rate estimation \citep{jlly2024}.

Our analysis reveals that this presumed geometric alignment is not always globally valid for the continuous SMC systems. The entry ratio may, in fact, bifurcate and deviate markedly from the slope of the sliding surface, depending critically on the initial conditions and controller parameters. It is important to clarify at the outset that the present paper does not challenge the invariance property of the sliding motion after the state has reached the sliding surface. It is well known that, once on the manifold, the phase-plane trajectory indeed obeys the structure of the designed sliding surface and the state ratio is uniquely determined. Rather, our focus is on the geometric manner in which the trajectory enters the sliding surface when continuous SMC is applied.

The revelation that this entry is not universally smooth has direct practical consequences: in hardware-in-the-loop emulations and digital implementations of SMC, the entry geometry affects the dynamic range and quantization accuracy of the controller near the switching surface, and may explain unexpectedly suboptimal convergence speeds observed in practice \citep{ZHANG2025112125}. From a theoretical standpoint, this overlooked geometric nuance also carries profound implications for the broader family of terminal SMC. It has been a long-standing open problem to furnish a fully rigorous proof of finite-time stability for second-order Terminal SMC systems \citep{man1994}. A central difficulty lies in the fact that finite-time convergence proofs rely on precise bounds of the phase-plane trajectory during the terminal phase. Most critically, on the behavior of the trajectory in the vicinity of the sliding manifold. While existing Lyapunov-based arguments typically circumvent this issue by employing homogeneous inequalities, they often implicitly rely on a favorable (i.e., smooth) entry geometry. Our findings on the linear SMC prototype demonstrate that such smooth entry is, in general, non-global and parameter-dependent. Thus, the bifurcation phenomenon we reveal for the linear case serves as a foundational cautionary example, highlighting that any rigorous finite-time proof for Terminal SMC must explicitly account for the entry geometry rather than take smooth entry for granted.

In this paper, we challenge the entrenched smooth-entry assumption by conducting a thorough phase-plane analysis of the classical double-integrator system under linear SMC, as depicted in Fig.~\ref{fig:presentation}. We show that, for the classical linear SMC law and the specific sliding mode, the entry ratio bifurcates according to the controller gain, the sliding surface parameter, and the initial state. Consequently, the revealed geometric bifurcation is an intrinsic property of the closed-loop dynamics, not an artifact of a particular control implementation. We further extend the analysis to the broader class of PD controllers and show that the sliding-mode result emerges as a special case of a more general critical parameter boundary, thereby unifying the understanding of second-order linear feedback systems. The main contributions of this work are summarized as follows:

\begin{figure}
\centering
\includegraphics[width=0.67\textwidth, trim=0pt 0pt 0pt 0pt, clip]{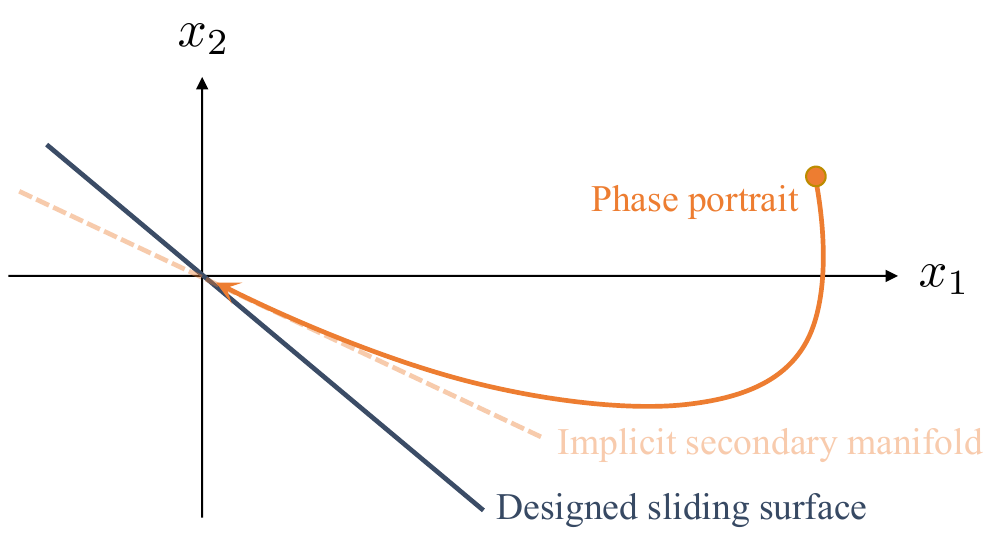}
\caption{Trajectory does not align with the designed sliding surface.}
\label{fig:presentation}
\end{figure}

 \begin{enumerate}
 	\item A novel bifurcation phenomenon in the asymptotic convergence of continuous SMC is uncovered, breaking the conventional tangential surface convergence consensus.
	\item Closed-form convergence ratio expressions and a unified classification framework are derived to judge the validity and failure of the classic smooth-entry assumption.
	\item The inherent correlation between the proposed bifurcation threshold and the critical damping boundary of PD control is revealed, guiding high-performance motion control design.
\end{enumerate}

The remainder of this paper is organized as follows. Section~\ref{sec:problem} formulates the problem about the smooth sliding entry. Section~\ref{sec:results} presents the bifurcation phenomenon of continuous SMC systems and the analytical characterization of the entry ratio, and extends the results to the general PD control framework. Section~\ref{sec:simu} provides numerical validations and phase-plane illustrations, and Section~\ref{sec:conclusion} concludes the paper with a discussion on future directions.

\section{Problem Formulation}
\label{sec:problem}
Consider a classical double-integrator linear system as
\begin{align}
\label{eq_ddcls_1}
\dot{x}_1(t) &= x_2(t) \nonumber \\
\dot{x}_2(t) &= u(t),
\end{align}
where $x_1, x_2 \in \mathbb{R}$ are the system states and $u \in \mathbb{R}$ is the control input. \footnote{For clarity, $x\in\mathbb{R}$ denotes a point in the state space, whereas $x(t)\in\mathbb{R}$ denotes the solution trajectory of the system.} This model describes a broad class of second-order physical systems, such as simplified mechanical positioning stages and vehicle dynamics. A linear sliding surface is defined as
\begin{align}
\label{eq_ddcls_2}
s(t) = x_1(t) + kx_2(t),
\end{align}
where $k > 0$ is a design parameter that determines the slope of the sliding surface in the phase plane. 

To rigorously characterize the terminal convergence geometry, we introduce the definition of smooth sliding entry.

\begin{definition}[Smooth sliding entry]
\label{def1}
The terminal convergence of the closed-loop system \eqref{eq_ddcls_1} is said to exhibit \emph{smooth sliding entry} if
\begin{align}
\lim_{t \to \infty} \frac{x_1(t)}{x_2(t)} = -k,
\end{align}
where $k$ is the slope of the prescribed sliding surface \eqref{eq_ddcls_2}. Otherwise, the convergence is said to exhibit \emph{non-smooth entry}, and the degree of misalignment is quantified by the indicator
\begin{align}
\Psi \triangleq \left| \arctan\left( \lim_{t \rightarrow \infty} \frac{x_1(t)}{x_2(t)} \right) + \arctan k \right|.
\end{align}
\end{definition}

\begin{rem}
According to Definition \ref{def1}, tangential smooth entry ($\Psi = 0$) indicates that the phase plane trajectory asymptotically aligns with the sliding surface, ensuring a seamless transition to the equilibrium. However, $\Psi \neq 0$ indicates that the trajectory converges to the origin along a path distinct from the sliding surface, revealing the emergence of bifurcation. This deviation quantifies the ``lack of smoothness'' during the terminal convergence phase.
\end{rem}

With Definition~\ref{def1}, we characterize the convergence behavior. Specifically, the following issues are investigated:
\begin{enumerate}
    \item To derive the exact analytical expression for the asymptotic state ratio $$\lim_{t \to \infty} \frac{x_1(t)}{x_2(t)}$$ and rigorously prove its bifurcating dependence on the controller parameter $k$ and the initial conditions.
    
    \item To demonstrate that the classical linear SMC system inherently exhibits two convergence modes governed by an implicit secondary manifold, and to provide its precise mathematical description.
    
    \item To extend the analysis to general PD controllers and identify the critical threshold that governs both the bifurcation behavior and the equivalence between PD and sliding mode control.
\end{enumerate}

To focus the analysis on convergence, we impose the following non-triviality assumption on initial conditions.

\begin{ass}
\label{ass1}
$|x_1(0)|+|x_2(0)|\neq 0$ holds.
\end{ass}

\begin{rem}
If the initial condition satisfies $x_1(0)=0$ and $x_2(0)=0$, then the system would perpetually satisfy $x_1(t) \equiv 0$ and $x_2(t) \equiv 0$ for all $t \geqslant 0$ under the given control, which constitutes a degenerate case for asymptotic analysis.
\end{rem}

\section{Main Results of SMC Bifurcation}
\label{sec:results}

This section establishes the analytical framework for the asymptotic state ratio under a general class of linear sliding mode control laws. We reveal that the terminal convergence behavior is not unique but exhibits a bifurcation phenomenon determined by the interplay between control parameters and initial conditions.

We introduce a general linear control input as
\begin{align}
	\label{eq_control_law_3}
	u(t) = -\frac{1}{k} \left(x_2(t) + hs(t)\right),
\end{align}
where $h > 0$ is a reaching rate parameter that generalizes the reaching behavior. This control law encompasses two classical cases: $h = 1$ corresponds to the standard linear reaching law, and $h = k$ yields a parameter-dependent reaching rate. Based on the control law \eqref{eq_control_law_3}, we have the following theorem.

\begin{thm}
	\label{thm_3}
	Consider the closed-loop system \eqref{eq_ddcls_1} the control law \eqref{eq_control_law_3} and suppose that Assumption \ref{ass1} holds. Then, the equilibrium at the origin is globally asymptotically stable. Moreover, the system exhibits a bifurcation behavior characterized by the asymptotic ratio
	\begin{equation}
		\label{eq_limit_ratio_3}
		\Psi 
		= \begin{cases} 
			0, & \mathrm{if}~ \mathcal{C}_1 \\
			\psi_{\mathrm{nz}}, & \mathrm{if}~ \mathcal{C}_2
		\end{cases}
	\end{equation}
	where $\psi_{\mathrm{nz}}$ is some nonzero scalar, and the convergence regimes partition the parameter and initial state space as
	\begin{itemize}
		\item $\mathcal{C}_1: 
			\begin{aligned}[t]
				&(0 < kh < 1 \text{ and } s(0) = 0) \\
				& \text{or } (kh > 1 \text{ and } hx_1(0)+x_2(0) \neq 0) \\
				&\text{or } (kh=1);
			\end{aligned}$
		\item $\mathcal{C}_2: 
			\begin{aligned}[t]
				&(0 < kh < 1 \text{ and } s(0) \neq 0) \\
				&\text{or } (kh > 1 \text{ and } hx_1(0)+x_2(0) = 0).
			\end{aligned}$
	\end{itemize}
\end{thm}

\begin{proof}
	The proof proceeds in three parts: stability analysis, Asymptotic State Ratio Analysis for $kh\neq 1$, and Asymptotic State Ratio Analysis for $kh=1$.
	
	\textbf{Part 1: Stability Analysis.}
	Consider the Lyapunov function candidate 
		\begin{align*}
			V(t) = \frac{1}{2} s^2(t).
		\end{align*}
		Its derivative along the closed-loop trajectories is
	\begin{align}
	\label{eq6_add}
		\dot{V}(t) = s\left(x_2(t) + k u(t)\right),
	\end{align}
	where \eqref{eq_ddcls_1} and \eqref{eq_ddcls_2} are utilized. Substituting the control law \eqref{eq_control_law_3} into \eqref{eq6_add} yields
	\begin{align}
		\label{eq_8}
		\dot{V}(t) = -h s^2(t).
	\end{align}
	By LaSalle's invariance principle, all trajectories converge to the largest invariant set contained in $\{s=0\}$. On $s=0$, the control reduces to $$u(t) = -\frac{x_2(t)}{k},$$ giving $$\dot{x}_2 = -\frac{x_2(t)}{k} \text{and} x_1(t) = -k x_2(t)$$, which asymptotically approaches the origin. Hence the origin is globally asymptotically stable.

	\textbf{Part 2: Asymptotic State Ratio Analysis for $kh\neq 1$.}
	To analyze the terminal convergence behavior, we introduce the new coordinates
	\begin{align}
	\label{eq8_add}
		\sigma_1(t) &= h x_1(t) + x_2(t)
		\nonumber\\
		\sigma_2(t) &= s(t).
	\end{align}
	Based on this coordinate transformation, we decouple the system \eqref{eq_ddcls_1} into two independent first-order linear differential equations as
	\begin{align}
		\dot{\sigma}_1(t) 
		&= -\frac{1}{k}\sigma_1(t)
		\nonumber
		\\
		\dot{\sigma}_2(t) &=  -h \sigma_2(t).
	\end{align}
	Their solutions are
	\begin{align}
	\label{eq10_add}
		\sigma_1(t) &= \sigma_1(0) e^{-\frac{t}{k}}
		\nonumber
		\\ 
		\sigma_2(t) &= \sigma_2(0) e^{-h t}.
	\end{align}
	Since $kh\neq 1$, the inverse transformation is well defined and yields
	\begin{align}
		x_1(t) &= \frac{k \sigma_1(t) - \sigma_2(t)}{kh - 1}
		\nonumber
		\\
		x_2(t) &= \frac{h \sigma_2(t) -  \sigma_1(t)}{kh - 1}.
	\end{align}
	Consequently,
	\begin{align}
	\label{eq12_add}
		\frac{x_1(t)}{x_2(t)} = \frac{k \sigma_1(t) - \sigma_2(t)}{h \sigma_2(t) - \sigma_1(t)}.
	\end{align}
	Substituting \eqref{eq10_add} into \eqref{eq12_add}, we have
	\begin{align}
		\frac{x_1(t)}{x_2(t)} = \frac{k \sigma_1(0) e^{-\frac{t}{k}} - \sigma_2(0) e^{-h t}}{h \sigma_2(0) e^{-h t} - \sigma_1(0) e^{-\frac{t}{k}}}.
	\end{align}
	Dividing both numerator and denominator by $e^{-h t}$ yields
	\begin{align}
	\label{eq15_add}
		\frac{x_1(t)}{x_2(t)} = -  \frac{\sigma_2(0) - k \sigma_1(0) \rho(t)}{h \sigma_2(0) - \sigma_1(0) \rho(t)},
	\end{align}
	where
	\begin{align*}
	\rho(t) \triangleq e^{\frac{kh-1}{k}t}.
	\end{align*}
	The asymptotic behavior is determined by the sign of $kh-1$. Next, the following cases are considered.
	
	\begin{itemize}	
		\item {Case of $0 < kh < 1$ and $\sigma_2(0) = 0$:} If $\sigma_2(0) = 0$ holds, then $\sigma_1(0) \neq 0$ according to \eqref{eq8_add} and Assumption \ref{ass1}. Hence,
		\begin{align}
			\lim_{t \rightarrow \infty} \frac{x_1(t)}{x_2(t)} = -k,
		\end{align}
		which gives $\Psi= 0$. 
		
		\item {Case of $0 < kh < 1$ and $\sigma_2(0) \neq 0$:} If $\sigma_2(0) \neq 0$ and $kh - 1 < 0$, then $\sigma_1(0) \neq 0$ and $\lim_{t\to \infty}\rho(t) = 0$. As a result,
		\begin{align}
			\lim_{t \rightarrow \infty} \frac{x_1(t)}{x_2(t)} = -\frac{1}{h},
		\end{align}
		which gives $\Psi\neq 0$.
		
		\item {Case of $kh > 1$ and $\sigma_1(0) \neq 0$:} It is known that $kh > 1$ implies that $\rho(t)$ tends to positive infinity. This fact, together with $\sigma_1(0) \neq 0$, leads to
		\begin{align}
			\lim_{t \rightarrow \infty} \frac{x_1(t)}{x_2(t)} = -k,
		\end{align}
		which gives $\Psi= 0$.

		\item {Case of $kh > 1$ and $\sigma_1(0) = 0$:} If $\sigma_1(0) = 0$ holds, then $\sigma_2(0) \neq 0$ according to \eqref{eq8_add} and Assumption \ref{ass1}. This fact, together with \eqref{eq15_add}, implies
		\begin{align}
			\lim_{t \rightarrow \infty} \frac{x_1(t)}{x_2(t)} = -\frac{1}{h}, 
		\end{align}
		which gives $\Psi\neq 0$.
		\end{itemize}
		
		\textbf{Part 3: Asymptotic State Ratio Analysis for $kh = 1$.} When $kh = 1$, we have $\sigma_1(t) = h \sigma_2(t)$, and the decoupling coordinate transformation used above becomes singular. Solving the system \eqref{eq_ddcls_1} with the aid of \eqref{eq10_add} gives 
		\begin{align*}
		\dot{x}_1(t) &= x_2(t)
		\\
		\dot{x}_2(t) &= -\frac{1}{k} x_2(t) - \frac{1}{k^2} \sigma_2(0) e^{-\frac{t}{k}},
		\end{align*}
		and consequently,
		\begin{align}
		x_1(t) &= \left({x_1(0) + \frac{1}{k} \sigma_2(0)t}\right)e^{-\frac{t}{k}} \nonumber \\
		x_2(t) &= \left({x_2(0) - \frac{1}{k^2} \sigma_2(0)t}\right){e^{-\frac{t}{k}}}.
		\end{align}		
		If $\sigma_2(0) = 0$, then $x_1(0) = -k x_2(0)$ by \eqref{eq8_add}, which implies
		\begin{align}
			\lim_{t \rightarrow \infty} \frac{x_1(t)}{x_2(t)} &=
			\frac{x_1(0)}{x_2(0)}
			\nonumber
			\\&= -k. 
		\end{align}
		\noindent If $\sigma_2(0) \neq 0$, then
		\begin{align}
			\lim_{t \rightarrow \infty} \frac{x_1(t)}{x_2(t)} &=
			\lim_{t \rightarrow \infty} \frac{\frac{1}{k} \sigma_2(0)t}{-\frac{1}{k^2} \sigma_2(0)t}
			\nonumber
						\\&
						=-k. 
		\end{align}
		Hence $\Psi=0$ holds for all initial conditions when $kh = 1$.

The above analysis establishes \eqref{eq_limit_ratio_3}, demonstrating that the system undergoes a bifurcation in its asymptotic behavior. This completes the proof of Theorem~\ref{thm_3}.
\end{proof}
\begin{rem}
\label{rem_completeness}
Since Assumption \ref{ass1} excludes the equilibrium as an initial state, the conditions $s(0)=0$ and $h x_1(0)+x_2(0)=0$ are mutually exclusive for any $kh \in (0,1) \cup (1,\infty)$. Consequently, the regimes $\mathcal{C}_1$ and $\mathcal{C}_2$ cover all possible convergence scenarios for $k > 0$ and $h > 0$.
\end{rem}

\begin{rem}
The bifurcation of the state ratio is an intrinsic property of the control structure. The parameters $k$ and $h$ serve as the ``modal weights.'' This shift is what forces the system trajectory to switch its terminal alignment, leading to the two convergence behavior observed in Theorem \ref{thm_3}.
\end{rem}

\begin{rem}
The bifurcation directly affects two key performance metrics: convergence speed and transient smoothness. When $kh\neq 1$, the effective decay rate of the states is either $\frac{1}{k}$ or $h$. Moreover, $\psi_{\mathrm{nz}}\neq 0$ shows that the trajectory bypasses the designed sliding surface. Thus, the bifurcation provides a theoretical basis for parameter selection and initial condition planning in high-precision motion control.
\end{rem}

It is worth noting that the sliding mode control law \eqref{eq_control_law_3} can be viewed as a special case of a PD controller of the form \eqref{eq_control_law_1}, where the proportional and derivative gains are not independent but instead satisfy a particular algebraic relation imposed by the sliding surface design. A natural extension is to consider the general PD control law
\begin{align}
\label{eq_control_law_1}
u(t) = -\alpha x_1(t) - \beta x_2(t),
\end{align}
where $\alpha>0$ and $\beta>0$.

\begin{thm}
\label{cor_1}
Consider the closed-loop system \eqref{eq_ddcls_1} under the control law \eqref{eq_control_law_1} and suppose that Assumption \ref{ass1} holds. Then the equilibrium at the origin is globally asymptotically stable. Moreover, the asymptotic behavior of the system depends on the controller gains as follows:
\begin{itemize}
    \item If $\beta > 2\sqrt{\alpha}$, the system exhibits a bifurcation behavior, i.e., $\Psi = 0$ for some initial conditions and $\Psi = \psi_{\mathrm{nz}}$ for others, where $\psi_{\mathrm{nz}}$ is a nonzero scalar.
    \item If $\beta = 2\sqrt{\alpha}$, the system exhibits no bifurcation, and $\Psi = 0$ holds for all initial conditions.
\end{itemize}
\end{thm}

\begin{proof}
The proof follows directly from Theorem~\ref{thm_3} via a straightforward controller parameter transformation and is therefore omitted for brevity.
\end{proof}


\begin{rem}
Extending the linear SMC \eqref{eq_control_law_3} to the general PD structure \eqref{eq_control_law_1} not only preserves the global asymptotic stability of the origin but also reveals a richer bifurcation structure parameterized by the controller gains $\alpha$ and $\beta$. 
This extension uncovers the threshold $\beta = 2\sqrt{\alpha}$ as the critical sliding surface boundary: when $\beta > 2\sqrt{\alpha}$, the system exhibits a bifurcated sliding mode with the asymptotic regime depending on the initial conditions; when $\beta = 2\sqrt{\alpha}$, the sliding mode is non-bifurcated and $\Psi = 0$ holds for all initial conditions; when $\beta < 2\sqrt{\alpha}$, no classical linear sliding mode corresponding to \eqref{eq_control_law_3} exists.
\end{rem}



\section{Simulation of Four Asymptotic Cases}
\label{sec:simu}

In this section, numerical simulations are conducted to investigate the dynamic evolution characteristics of the closed-loop system formulated in \eqref{eq_ddcls_1}, where the control law designed in \eqref{eq_control_law_3} is adopted as the core control strategy for all simulation scenarios. The primary objective of the numerical study is to quantitatively and qualitatively verify the theoretical asymptotic properties of the system state ratio derived in \eqref{eq_limit_ratio_3}, with a particular focus on validating the distinct dynamic behaviors corresponding to the four predefined asymptotic cases.

\begin{table}[!hb]
	\centering
	\caption{Simulation Parameters for Four Cases}
	\label{tab:four_cases_params}
	\begin{tabular}{@{}lccccc@{}}
		\toprule
		\textbf{Case} & {Description} & $k$ & $h$ & $x_1(0)$ & $x_2(0)$ \\
		\midrule
		$1$ & $0<kh<1$, $x_1(0)+kx_2(0) = 0$ & $2.0$ & $0.4$  & $2.0$ & $-1.0$   \\
		$2$ & $kh>1$, $hx_1(0)+x_2(0) \neq 0$ & $2.0$ & $0.6$  & $2.0$ & $1.0$  \\
		$3$ & $0<kh<1$, $x_1(0)+kx_2(0) \neq 0$ & $2.0$ & $0.4$  & $2.0$ & $1.0$ \\
		$4$ & $kh>1$, $hx_1(0)+x_2(0) = 0$ & $2.0$ & $0.6$ & $2.0$ & $-1.2$ \\
		\bottomrule
	\end{tabular}
\end{table}

All simulation parameters and initial state conditions utilized in the numerical experiments are summarized in Tab.~\ref{tab:four_cases_params}. These configurations are elaborately selected to strictly satisfy the prerequisite theoretical constraints for each of the four asymptotic cases, ensuring that the simulation scenarios can fully reproduce the dynamic regimes analyzed in the theoretical derivation. Such a parameter setting scheme guarantees the consistency between numerical verification and theoretical analysis, and enables a reliable validation of the proposed theoretical results.

The theoretical conclusions presented in the previous sections are comprehensively validated via four cases of comparative numerical simulations, whose dynamic responses are visualized in Figs.~\ref{fig:three_cases_phase1}, \ref{fig:three_cases_phase2}, \ref{fig:three_cases_phase3}, and \ref{fig:three_cases_phase4}. Specifically, Figs.~\ref{fig:sub1}, \ref{fig:sub4}, \ref{fig:sub7}, and \ref{fig:sub10} depict the phase plane trajectories of the system \eqref{eq_ddcls_1} in the state space under the parameter and initial condition configurations corresponding to the four cases. For the first two asymptotic cases, the system phase trajectories gradually approach the system equilibrium point with a tangential trend toward the predefined sliding surface \(s=0\). This typical dynamic characteristic corresponds to the conventional smooth sliding mode entry behavior, which is consistent with the classic sliding mode dynamic mechanism. In sharp contrast, the last two cases exhibit entirely different convergence characteristics, where the system phase trajectories converge along an implicit secondary manifold rather than the predefined sliding surface. Notably, the fourth case presents the most distinct geometric dynamic feature: the steady-state trajectory fails to align with the predefined sliding manifold \(s=0\), but maintains a constant and noticeable deflection angle relative to the manifold, and finally converges asymptotically along the inherent secondary slope. The above intuitive simulation observations solidly demonstrate that the terminal dynamic behavior of the system state trajectory is dominated by the dominant exponential convergence mode of the system, rather than constrained by the artificially predefined sliding surface \(s=0\).

\begin{figure}
	\centering
	\begin{subfigure}{1\textwidth}
		\centering
		\includegraphics[width= 0.67\textwidth, trim=0 0 0 0, clip]{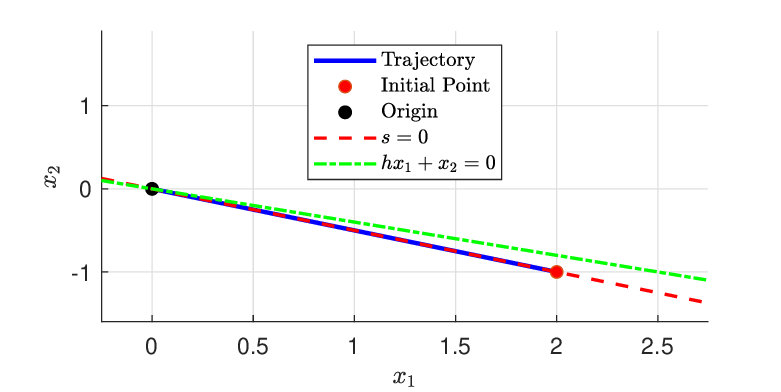}
		\caption{Phase portrait and state trajectories.}
		\label{fig:sub1}
	\end{subfigure}
	
	
	\begin{subfigure}{1\textwidth}
		\centering
		\includegraphics[width= 0.67\textwidth, trim=0 0 0 0, clip]{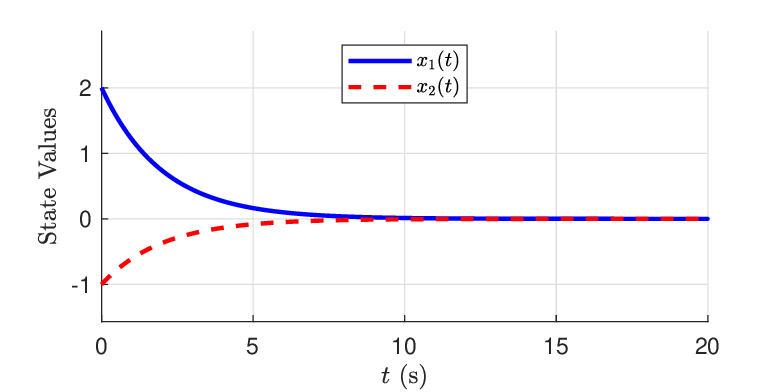}
		\caption{The evolution of two states.}
		\label{fig:sub2}
	\end{subfigure}
	
	
	\begin{subfigure}{1\textwidth}
		\centering
		\includegraphics[width= 0.67\textwidth, trim=0 0 0 0, clip]{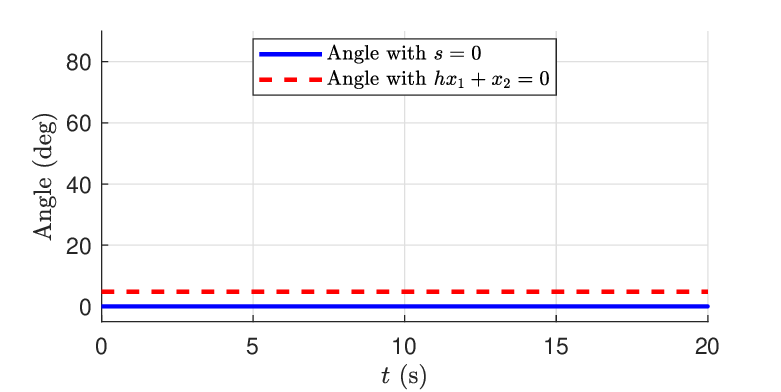}
		\caption{Angular relationship between phase trajectories and reference lines.}
		\label{fig:sub3}
	\end{subfigure}
	
	\caption{Simulation Results for Case 1.}
	\label{fig:three_cases_phase1}
\end{figure}

\begin{figure}
	\centering
	\begin{subfigure}{1\textwidth}
		\centering
		\includegraphics[width= 0.67\textwidth, trim=0 0 0 0, clip]{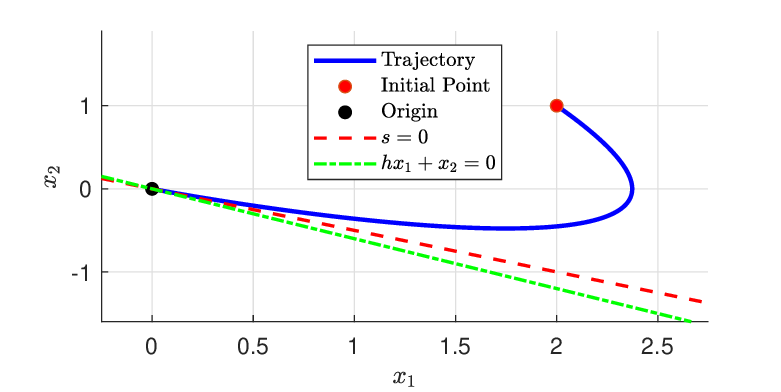}
		\caption{Phase portrait and state trajectories.}
		\label{fig:sub4}
	\end{subfigure}
	
	
	\begin{subfigure}{1\textwidth}
		\centering
		\includegraphics[width= 0.67\textwidth, trim=0 0 0 0, clip]{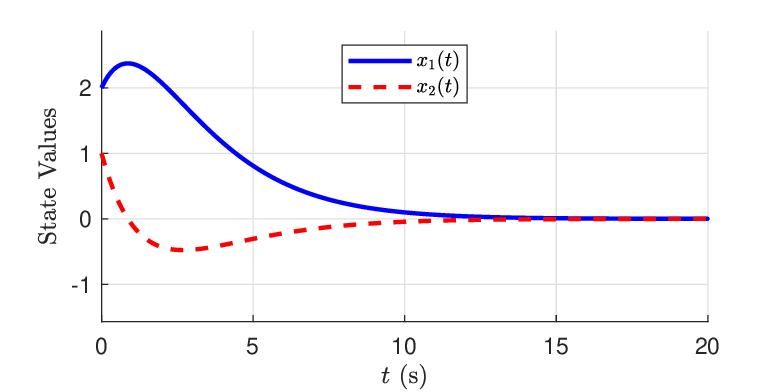}
		\caption{The evolution of two states.}
		\label{fig:sub5}
	\end{subfigure}
	
	
	\begin{subfigure}{1\textwidth}
		\centering
		\includegraphics[width= 0.67\textwidth, trim=0 0 0 0, clip]{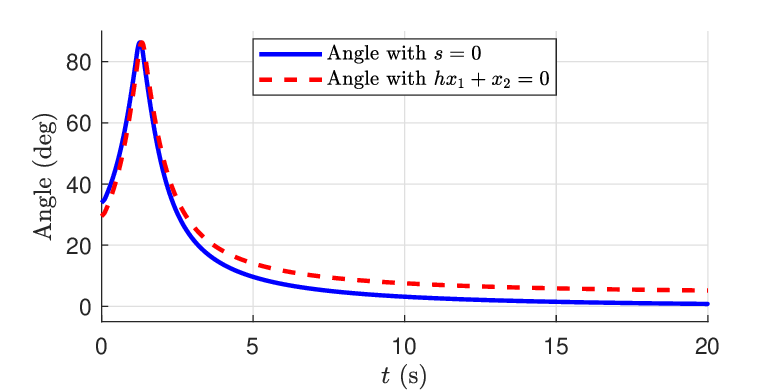}
		\caption{Angular relationship between phase trajectories and reference lines.}
		\label{fig:sub6}
	\end{subfigure}
	
	\caption{Simulation Results for Case 2.}
	\label{fig:three_cases_phase2}
\end{figure}

\begin{figure}
	\centering
	\begin{subfigure}{1\textwidth}
		\centering
		\includegraphics[width= 0.67\textwidth, trim=0 0 0 0, clip]{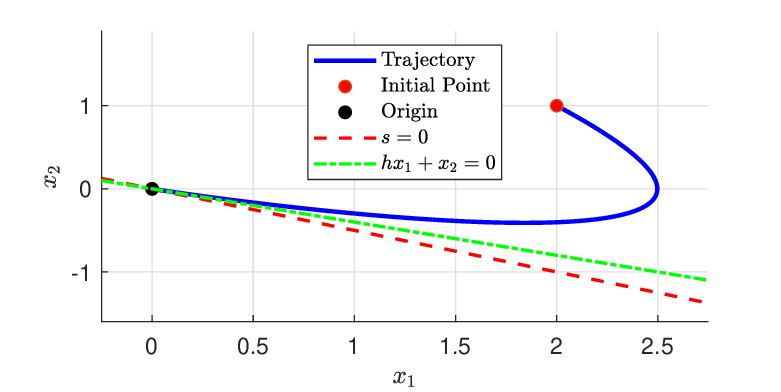}
		\caption{Phase portrait and state trajectories.}
		\label{fig:sub7}
	\end{subfigure}
	
	
	\begin{subfigure}{1\textwidth}
		\centering
		\includegraphics[width= 0.67\textwidth, trim=0 0 0 0, clip]{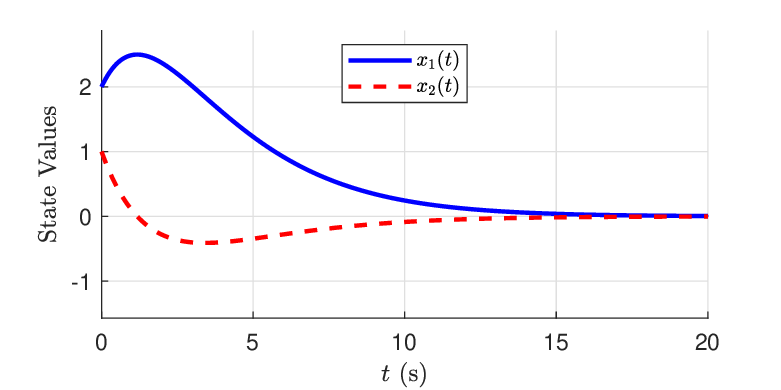}
		\caption{The evolution of two states.}
		\label{fig:sub8}
	\end{subfigure}
	
	
	\begin{subfigure}{1\textwidth}
		\centering
		\includegraphics[width= 0.67\textwidth, trim=0 0 0 0, clip]{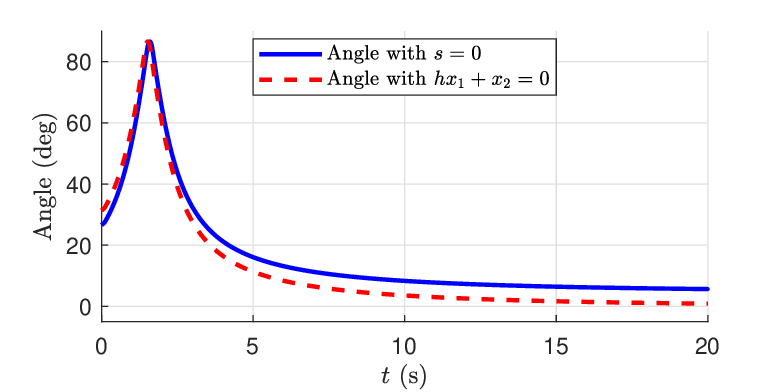}
		\caption{Angular relationship between phase trajectories and reference lines.}
		\label{fig:sub9}
	\end{subfigure}
	
	\caption{Simulation Results for Case 3.}
	\label{fig:three_cases_phase3}
\end{figure}

\begin{figure}
	\centering
	\begin{subfigure}{1\textwidth}
		\centering
		\includegraphics[width= 0.67\textwidth, trim=0 0 0 0, clip]{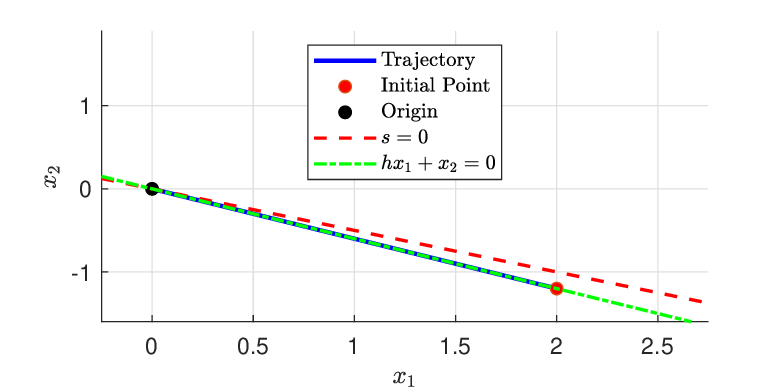}
		\caption{Phase portrait and state trajectories}
		\label{fig:sub10}
	\end{subfigure}
	
	
	\begin{subfigure}{1\textwidth}
		\centering
		\includegraphics[width= 0.67\textwidth, trim=0 0 0 0, clip]{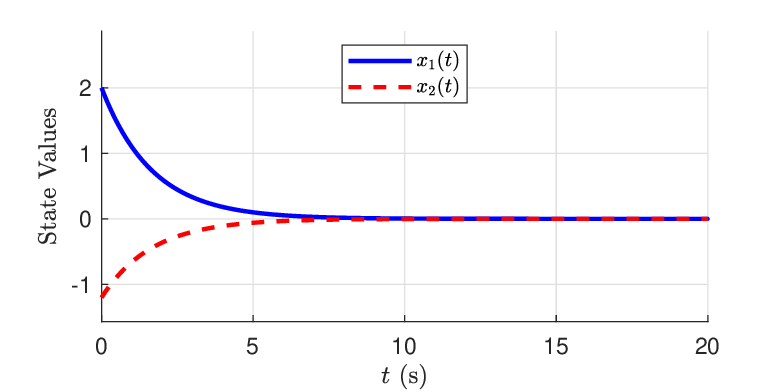}
		\caption{The evolution of two states.}
		\label{fig:sub11}
	\end{subfigure}
	
	
	\begin{subfigure}{1\textwidth}
		\centering
		\includegraphics[width= 0.67\textwidth, trim=0 0 0 0, clip]{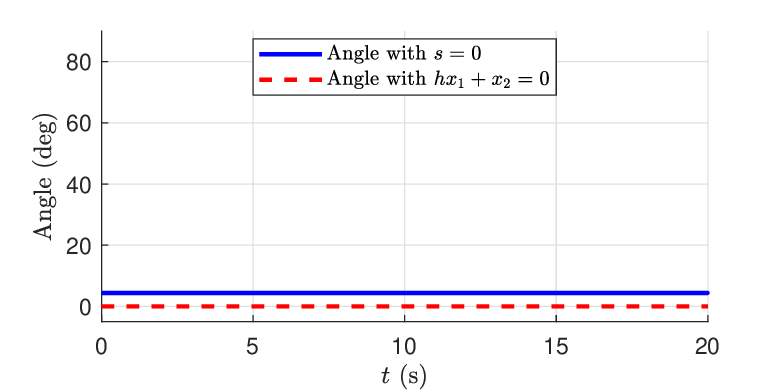}
		\caption{Angular relationship between phase trajectories and reference lines.}
		\label{fig:sub12}
	\end{subfigure}
	
	\caption{Simulation Results for Case 4.}
	\label{fig:three_cases_phase4}
\end{figure}

Furthermore, Figs.~\ref{fig:sub2}, \ref{fig:sub5}, \ref{fig:sub8}, and \ref{fig:sub11} illustrate the time-domain evolution curves of the system states \(x_1(t)\) and \(x_2(t)\) for all four simulation cases. It can be clearly observed that the system state variables \(x_1(t)\) and \(x_2(t)\) continuously decay and eventually converge to zero under all four scenarios. This time-domain response characteristic verifies that the closed-loop system maintains global asymptotic stability uniformly, irrespective of the differences in convergence pathways and asymptotic dynamic regimes among the four cases, which further confirms the stability conclusion of the theoretical analysis.

Figs.~\ref{fig:sub3}, \ref{fig:sub6}, \ref{fig:sub9}, and \ref{fig:sub12} present the time-domain variation of the indicator variable \(\Psi\) corresponding to the four simulation cases. The results show that \(\Psi=0\) is always satisfied throughout the dynamic evolution process for the first two cases, whereas \(\Psi\neq 0\) holds steadily for the latter two cases. This distinct difference directly reveals that the system cannot achieve smooth entry into the predefined sliding surface when the \(\mathcal{C}_2\) condition is satisfied. The essential cause of this phenomenon lies in the generation and dominance of the implicit secondary manifold in the closed-loop system dynamics, and this special manifold-induced dynamic deviation can also be clearly observed from the phase plane trajectory characteristics of the fourth case in Figs.~\ref{fig:sub1} and \ref{fig:sub4}.

In addition, the theoretical asymptotic state ratios of the four cases predicted by Theorem~\ref{thm_3} are \(-2.0000\), \(-2.0000\), \(-2.5000\), and \(-1.6667\), respectively, as quantified in Tab.~\ref{tab:four_cases_params2}. The theoretical predicted values and the numerical state ratios extracted from closed-loop simulation results are comparatively listed in Tab.~\ref{tab:four_cases_params}. It is observed that the numerical simulation results are in excellent agreement with the theoretical calculation values, with the relative error of all four cases controlled below \(0.15\%\). The negligible minor discrepancies are solely attributed to the inherent finite step size of the numerical simulation algorithm. The high consistency between simulation and theoretical results effectively verifies the correctness and accuracy of the proposed Theorem~\ref{thm_3}, and demonstrates the validity of the theoretical analysis for the four asymptotic dynamic cases.

\begin{table}
	\centering
	\caption{Simulation Results for Four Cases}
	\label{tab:four_cases_params2}
	\begin{tabular}{@{}lccc@{}}
		\toprule
		\textbf{Case} & {Description} &  Theoretical  State Ratio Limit & simulated State Ratio Limit \\
		\midrule
		$1$ & $0<kh<1$, $x_1(0)+kx_2(0) = 0$  & $-2.0000$ & $-2.0025$ \\
		$2$ & $kh>1$, $hx_1(0)+x_2(0) \neq 0$  & $-2.0000$ & $-1.9975$ \\
		$3$ & $0<kh<1$, $x_1(0)+kx_2(0) \neq 0$  & $-2.5000$ & $-2.4975$ \\
		$4$ & $kh>1$, $hx_1(0)+x_2(0) = 0$  & $-1.6667$ & $-1.6642$ \\
		\bottomrule
	\end{tabular}
\end{table}

\section{Conclusion}
\label{sec:conclusion}

This paper has presented a detailed analysis of the convergence characteristics of linear SMC applied to second-order systems. By deriving analytical expressions for the state trajectories, we have obtained exact formulas for the asymptotic state ratio and characterized the conditions under which convergence occurs strictly along the designed sliding surface. A central contribution is the demonstration that the conventional intuition of a unique smooth sliding entry is fundamentally incomplete: the system exhibits a bifurcation in its convergence modes, wherein the terminal state ratio depends critically on the interplay among control gains, surface parameters, and initial conditions. The analysis has further been generalized to PD controllers, revealing that the bifurcation threshold. These findings offer both deeper theoretical insight and a practical criterion for controller tuning in finite or fixed-time control applications.

\printcredits

\section*{Declaration of Competing Interest}
The authors declare that they have no known competing financial interests or personal relationships that could have appeared to influence the work reported in this paper.

\section*{Acknowledgments}
This research was supported in part by Zhejiang Provincial Natural Science Foundation of China under Grant QN26F030021, in part by the Natural Science Basic Research Program of Shaanxi under Grant 2025JC-YBMS-718, and in part by the Fundamental Research Funds for the Central Universities.

\bibliographystyle{elsarticle-num} 

\bibliography{references}

\end{document}